\documentclass[11pt, onecolumn, draftclsnofoot]{IEEEtran}
\usepackage{amsmath,amssymb}
\usepackage[dvips]{graphicx}
\usepackage{amsfonts}
\usepackage[mathscr]{eucal}
\usepackage{latexsym}
\usepackage{amsthm}
\usepackage{exscale}
\usepackage[mathscr]{eucal}
\usepackage{bm}
\usepackage[dvipsnames]{color}
\usepackage{cases}
\usepackage{epsfig}
\usepackage[center,small]{caption}
\usepackage{algorithm}
\usepackage{algorithmic}
\usepackage[verbose,nospace,sort]{cite}
\usepackage{tabularx}
\usepackage{multirow}
\usepackage{balance}
\usepackage{url}
\scrollmode

\newtheorem{lemma}{Lemma}

\newtheorem{example}{Example}

\IEEEoverridecommandlockouts

\begin{document}
\title{Two-Phase Cooperative Broadcasting Based on Batched Network Code}
\author{Xiaoli~Xu, Praveen Kumar M. Gandhi, Yong~Liang~Guan,  and Peter~Han~Joo~Chong
\thanks{Part of this work will be presented in IEEE International Conference on Communications (ICC), London, UK, June 2015.}
\thanks{The authors are with the School of Electrical and Electronic Engineering, Nanyang Technological University, Singapore 639798
(email: \{xuxiaoli, meenakshi, eylguan, ehjchong\}@ntu.edu.sg).}
}

\maketitle

\begin{abstract}
In this paper, we consider the wireless broadcasting scenario with a source node sending some common information to a group of closely located users, where each link is subject to certain packet erasures. To ensure reliable information reception by all users, the conventional approach generally requires repeated transmission by the source until all the users are able to decode the information, which is inefficient in many practical scenarios.  In this paper, by exploiting the close proximity among the users,  we propose a novel two-phase wireless broadcasting protocol with  user cooperations based on an efficient batched network code, known as batched sparse (BATS) code.  In the first phase, the information packets are encoded into batches with BATS encoder and sequentially broadcasted by the source node until certain terminating criterion is met. In the second phase,  the users cooperate with each other by exchanging the network-coded information via peer-to-peer (P2P) communications based on  their respective received packets.  A fully distributed and light-weight scheduling algorithm is proposed to improve the efficiency of the P2P communication in the second phase. The performance of the proposed two-phase protocol is analyzed and the channel rank distribution at the instance of decoding is derived, based on which the optimal BATS code is designed. Simulation results demonstrate that the proposed protocol significantly outperforms the existing schemes.  Lastly, the performance of the proposed scheme is further verified via testbed experiments.
\end{abstract}


\section{Introduction}
Wireless broadcasting, by which some  common information is transmitted from a source node to a set of receiver nodes through wireless channels, has a wide range of applications \cite{Zoman1994}, such as satellite communication, video streaming, and file distribution. The main design criterion for wireless broadcasting systems is to ensure reliable information reception by all users, who may experience quite different channel conditions due to channel fading, interference and/or congestions. The two most common approaches for ensuring reliable broadcasting are retransmission and coding \cite{Luby2006}. With the simple ``repeat request-retransmission" scheme, the source retransmits the lost packets upon receiving a negative acknowledgement (NAK) from any of the receivers. Although simple for implementation, this scheme typically results in poor bandwidth efficiency. On the other hand, the coding based approach, such as forward erasure correction coding, though more efficient, usually incurs high encoding/decoding complexity and severe delays. More recently, network coding based schemes have been proposed  to improve the efficiency of retransmission schemes \cite{Nguyen2008,Fragouli2006}. However, such schemes rely heavily on the prompt and accurate feedbacks from all the receivers, which are difficult to be achieved in practical communication systems.

All the aforementioned broadcasting schemes assume that there is no cooperation among the receivers, and hence reliable broadcasting can only be achieved via source retransmissions, either with or without coding. In this paper, we consider the scenario where the source node is intended to broadcast some common information to a group of users that are closely located within a small region. Such a setup models various practical communication scenarios, e.g., video streaming from a base station to a group of nearby  mobile users in cellular networks,  communication from a source user to a squadron of destination users in ad-hoc networks, etc. For such scenarios, the secondary channels between the destination users are usually  more reliable than the primary channels from the source node, due to the shorter distance.  By exploiting this fact, we propose a two-phase cooperative broadcasting scheme based on batched sparse (BATS) code, with limited source broadcasting in the first phase and network coded peer-to-peer (P2P) packet exchange in the second phase to achieve reliable decoding by all users.

Batch sparse (BATS) code is a joint fountain code and network code, first proposed in \cite{Yang11} for achieving the optimal throughput of wireless erasure networks with finite coding length \cite{Xu2014}. Compared with the pure fountain code, BATS code achieves higher throughput by allowing the intermediate nodes to re-encode the packets. Compared with random linear network coding (RLNC) \cite{Ho06}, BATS code requires smaller buffer size and has efficient decoding method based on belief propagation \cite{Yang2014}.  Furthermore, since network coding is only performed within the batch, the overhead used to trace network coding coefficients is much smaller compared with RLNC.

In the first phase of the proposed protocol, the information packets are encoded into batches with BATS encoder. These encoded batches are sequentially broadcasted by the source node until certain terminating criterion, which may be designed to minimize the number of \emph{source} transmissions or  the total number of transmissions, is met. For the first design objective, the source can stop transmission immediately when the user group, if allowed to decode cooperatively,  is able to recover the file, although each individual user still cannot decode the information based on their own received packets. For the second design objective, the number of batches sent by the source is optimized so that the total number of transmissions in both phases is minimized.

In the second phase, the users help each other by broadcasting to their peers via  P2P communications based on their respective received packets from phase 1.  P2P erasure repair  has been previously studied for wireless video broadcasting in  Multimedia Broadcast/Multicast Service (MBMS) applications \cite{Sanigepalli2006}, where various scheduling schemes have been proposed. Later, an adaptive scheduling scheme was proposed in \cite{BOPPER} to  improve the efficiency of the P2P repair. Given the global state information of all users, a cooperative P2P repair (CPR) problem was formulated in \cite{Cheung2006}, which has been proved to be NP-hard. A suboptimal distributed CPR algorithm was proposed in \cite{CPR}. However, the proposed scheme still heavily relies on the exchange of perfect control information, which is difficult to be achieved when the inter-user links are lossy. Besides, the efficiency of the retransmission in this scheme is relatively low.  The CPR algorithm in \cite{CPR} has also been extended to network coded CPR (NC-CPR) in \cite{DCPR} by applying random linear network coding\cite{Ho06} in P2P communications. A distributed scheduling scheme was proposed in \cite{DCPR} based on the intuition that the user with more innovative packets should transmit earlier. However, when the number of users and the number of packets increase, the transmission overhead, which includes transmission of the network coding coefficients and control information, may overwhelm the gain of the proposed scheme. By integrating the idea of NC-CPR \cite{DCPR} and rarest first scheduling \cite{Legout2006}, a light-weight peer scheduling algorithm, termed ``cooperative Peer-to-peer Information Exchange (PIE)", was proposed in \cite{Fan2010}. Furthermore, when the users are not fully connected, a cluster based repair,  where the users are grouped into clusters with one user assigned as the cluster head (CH), have been shown to be more efficient than traditional P2P repair \cite{CDS}. The CH collects the information packets from its cluster members (CM), which are then exchanged with other CHs via P2P communications.  The tradeoff between the intra-cluster and inter-cluster repairs was studied in \cite{CDSextension}. To reduce the network coding overhead, an XOR network coding scheme was proposed to replace the RLNC for P2P repair in \cite{XORRepair}. 

All the existing schemes \cite{Sanigepalli2006,BOPPER,Cheung2006,CPR,DCPR,Legout2006,Fan2010,CDS,CDSextension,XORRepair} assume that the inter-user channels used in the second phase are lossless, so that some state information can be reliably exchanged before the starting of the P2P communications; otherwise, their performance may degrade severely if the state/control information is lost. In contrast, the proposed scheme in this paper is fully distributed, without requiring any state information exchange, and hence can be applied to networks with lossy links. Specifically, in our proposed scheme, each user estimates the ``usefulness" of sending a coded packet from a particular batch based on the number of packets  received during the first phase and its own transmission history during the second phase. In general, the more packets  received by a certain batch in phase 1, the more likely that a network coded packet generated from this batch is useful for its peers. Furthermore, for a given batch, the usefulness of its packets decreases with its transmissions. The usefulness matrices are generated distributively by each user at the end of phase 1, based on which a queue of batch ID with descending usefulness is created. When the user has a chance to transmit, a network coded packet generated from the front-most batch in the queue is broadcasted first. Phase 2 is complete when all the users are able to decode the file.

With a good BATS code, the user should be able to decode the file with a small overhead, e.g., a file containing $K$ packets should be decoded from $(1+\eta)K$ received packets with $\eta\ll 1$.  However, the performance of BATS code is largely dependent on a pre-defined degree distribution.  The optimal degree distribution can be obtained as a function of the channel rank distribution \cite{Yang13}. In wireless erasure networks with fixed topology, such as those considered in \cite{Yang11} and \cite{Dong2014}, the channel rank distribution can be obtained based on the erasure probability of each link. However, in P2P networks, the channel rank distribution is also affected by the communication protocol and the stopping time.  In this paper, we analyze the transmit efficiency of the proposed cooperative broadcast protocol and derive the resulting channel rank distribution at the instance of decoding, based on which a good BATS code is designed.

Simulation results show that the proposed two-phase protocol achieves highly reliable broadcasting with less number of transmissions, compared with the traditional single-phase transmissions and the existing cooperative broadcast schemes \cite{Fan2009,Fan2010}. Moreover, since a large number of transmissions are shifted from the source to the users, where less power is required per transmission, the proposed protocol is more power efficient. The performance of the proposed scheme is further validated experimentally with a 4-node testbed  based on the 802.11g Wi-Fi network, where the source and the receivers are connected in ad-hoc mode. It is found that the experimental results match very well with the analytical and simulation results. Furthermore, it is found that the BATS code overhead, designed based on the estimated channel rank distribution, is less than $1\%$, which is quite close to the optimal case with fixed channel rank distribution \cite{Yang13}.

The rest of this paper is organized as follows. Section~\ref{sec:model} introduces the system model. The proposed two-phase protocol is illustrated and analyzed in Section~\ref{sec:main}. The effectiveness and efficiency of the proposed protocol is evaluated in Section~\ref{sec:performance}. In Section~\ref{sec:experiment}, the testbed setup and the experimental results are presented. Finally, we conclude this paper in Section \ref{sec:conclusion}.

\emph{Notations}: Throughout this paper, random variables are represented by boldface upper-case letters and the probability of an event is denoted as $\Pr(\cdot)$. For a random variable $\mathbf X$, we use $\mathbb{E}[\mathbf X]$ to denote its expectation. Furthermore, scalars, vectors and matrices are represented by italic, boldface lower- and boldface upper-case letters, respectively. For a set $\mathcal{A}$, we use $|\mathcal{A}|$ to denote its cardinality and use $\mathcal{A}\setminus\mathcal{B}$ to denote set subtraction.

\section{System Model}\label{sec:model}
As shown in Fig.~\ref{F:model}, we consider a broadcasting scenario where a source node $s$  intends to send some common information to a group of $k$ users, which are closely located within a small region far away from the source node.    We assume that the obstruction and interference near the source may cause a common packet loss probability $p_0$ for all users. In addition, we further assume that the wireless link between the source to each user suffers from independent\footnote{The channels are assumed to be independent in rich scattering environment when the distance between any pair of users is larger than half of the wavelength.} and memoryless packet loss with probability $p_1$, and the links between the users have erasure probability $p_2$, where $0<p_0, p_1,p_2<1$.  Since the distance from the source to the user group is much larger than that between the users,  we assume that $p_2\leq p_1$.  A packet sent by the source node can be successfully received by each receiver with probability $(1-p_0)(1-p_1)$, and that sent by one of the users can be received by its peers with probability $1-p_2$.

\begin{figure}[htb]
\centering
\includegraphics[scale=0.8]{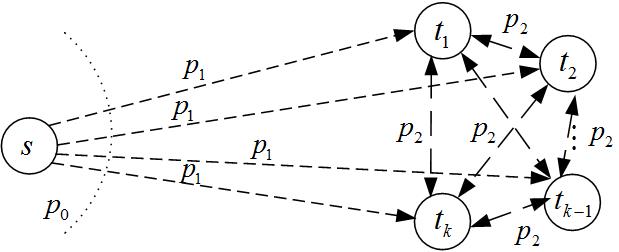}
\caption{Wireless broadcasting to a group of closely located users.}
\label{F:model}
\end{figure}

Intuitively, as the channel between the source and the user group is less reliable and much more power is required to compensate for the path loss over the long transmission distance from the source to the users, it is desirable  to minimize the number of transmissions by the source node by exploiting the more reliable P2P communication links between the users via user cooperation.

\section{Proposed Two-Phase Protocol with BATS Code}\label{sec:main}
In this section, we propose a two-phase transmission protocol based on BATS code and user cooperation to achieve reliable communication  for the scenarios shown in Fig.~\ref{F:model}. A BATS code consists of an outer code and an inner code, as shown in Fig.~\ref{F:Bats}. The outer code is an extension of the traditional fountain code to matrix form. Specifically, to apply BATS code, the source node first obtains a degree $d_i$ for $i$th batch by sampling a pre-designed degree distribution $\mathbf{\Psi}$, and then randomly picks $d_i$ distinct input packets to generate a batch of $M$ fountain-coded packets. The batches are then transmitted sequentially by the source. The inner code of BATS employs RLNC at the intermediate nodes, which corresponds to the users in Fig.~\ref{F:model}, and only packets within the same batch will be coded together. Hence, the network coding overhead is determined by the batch size $M$, which is usually  negligible compared with the packet length. Finally, the inner and outer codes are jointly decoded at the receiver using belief-propagation (BP) and inactivation decoding.

\begin{figure}[htb]
\centering
\includegraphics[scale=0.6]{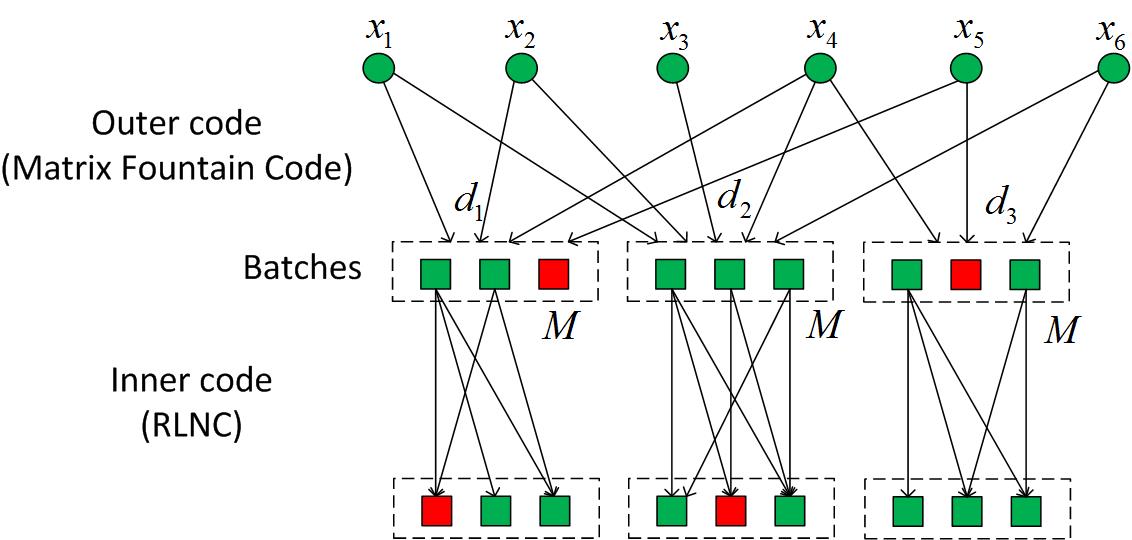}
\caption{Structure of BATS code. }
\label{F:Bats}
\end{figure}

Based on BATS code, we propose  the following two-phase transmission protocol:\\
\textbf{Phase 1}: The file is divided into $F$ packets, which are encoded with BATS code of batch size $M$ at the source node. The batches are broadcasted sequentially to the users until a stopping criterion is satisfied.\\
\textbf{Phase 2}: The users help each other by exchanging their respective received packets with i.i.d. erasures via temporally network-coded peer-to-peer (P2P) transmissions, until all the users can recover the file.


The efficiency of the BATS code is largely dependent on the degree distribution  $\mathbf{\Psi}$. In \cite{Yang13}, the optimal degree distribution is obtained by solving a linear optimization problem based on the file size and the channel rank distribution, which is assumed to be known before transmission. Since cooperative P2P repair will be used in the network of Fig.~\ref{F:model}, the channel rank distribution is affected not only by erasure probability, but also by scheduling algorithm. To design a BATS code for such network, we need to analyze the channel rank distribution observed by each user at the instance of decoding. In the following, the proposed protocol is discussed and analyzed based on two design objectives: i) minimizing the source transmissions; ii) minimizing the total number of retransmissions in both phases.

\subsection{Minimizing the source transmissions}\label{sec:phase1}
With optimal BATS code \cite{Yang13}, the file can be recovered if $(1+\eta)F$ packets are received, where $\eta\ll 1$ for moderate to large $F$. Denote by $\mathbf X$ the number of packets received by the user group.  To ensure that the user group can eventually recover the file via P2P transmissions with probability\footnote{$\varepsilon$ is a small positive number, which is set to $10^{-6}$ for the simulations in this paper.} no smaller than $(1-\varepsilon)$, we must have
\begin{align}
\Pr(\mathbf X\geq(1+\eta)F)\geq 1-\varepsilon. \label{eq:prob}
\end{align}

The probability that a packet is successfully received by at least one of the user equals to $(1-p_0)(1-p_1^k)$, or equivalently, the effective erasure probability is $\bar{p}\triangleq 1-(1-p_0)(1-p_1^k)$, where $k$ is the number of users in the group.  After $n$ batches (or equivalently $nM$ packets) have been sent by the source node, the number of received packets $\mathbf X$ is a random variable following binomial distribution $\mathcal{B}(nM,1-\bar{p})$. As $nM$ is usually  large, this binomial distribution can be approximated by the normal distribution $\mathcal{N}(\mu,\sigma^2)$ with $\mu=nM(1-\bar{p})$ and $\sigma^2=nM(1-\bar{p})\bar{p}$. Therefore, \eqref{eq:prob} can be approximated as
\begin{align}
Q\left(\frac{(1+\eta)F-\mu}{\sigma}\right)\geq 1-\varepsilon,\label{eq:thre}
\end{align}
where $Q(x)=\int_{x}^{\infty}\frac{1}{\sqrt{2\pi}}e^{-\frac{\tau^2}{2}}d\tau$ denotes the Gaussian $Q$ function. From \eqref{eq:thre},  the minimum value for $n$ can be obtained as
\begin{align}
n&\geq\frac{F'}{M(1-\bar{p})}-\frac{\alpha}{2M(1-\bar{p})}\left[\sqrt{4\bar{p}F'+\alpha^2\bar{p}^2}-\alpha \bar{p}\right]\nonumber\\
&\approx \frac{F'}{M(1-\bar{p})}-\frac{\alpha\sqrt{4\bar{p}F'}}{2M(1-\bar{p})},\label{eq:getn}
\end{align}
where $F'=(1+\eta)F$ and $\alpha=Q^{-1}(1-\varepsilon)$. The approximation in \eqref{eq:getn} is valid since $\alpha\bar{p}\ll F'$.  As a result, the minimum number of batches sent by the source is
\begin{align}
n_l=\left\lceil\frac{2F'-\alpha\sqrt{4\bar{p}F'}}{2M(1-\bar{p})}\right\rceil. \label{eq:nValue}
\end{align}

To minimize the source transmission, the source node stops transmission when the number of batches $n$ reaches the threshold $n_l$ given in \eqref{eq:nValue}. Since the expected number of packets received by each user is only $n_lM(1-p_0)(1-p_1)$, which is smaller than $(1+\eta)F$, the file is not yet recoverable by each individual user. In the second phase, the users help each other by broadcasting temporally network-coded packets generated from their respective received packets to ensure that all users can successfully recover the file eventually.


Since all the users are geographically separated, they have no knowledge on  what packets have been received by others during phase 1. Therefore, in phase 2, it is critical for each user to independently determine which packets it should send based on its own received packets.

Denote by $\mathcal{N}_i^j$ the set of packets received by user $t_j$ with batch index $i$, where $i\in\{1,...,n\},j\in\{1,...,k\}$.  Consider two typical users $t_j$ and $t_{j'}$.  With random linear network coding, a coded packet generated from $t_j$ for batch $i$ is useful for user $t_{j'}$ if $\mathcal{N}_i^j\setminus \mathcal{N}_i^{j'}\neq\emptyset$. Furthermore, if $|\mathcal{N}_i^j\setminus \mathcal{N}_i^{j'}|=m$, $m$ useful packets for batch $i$ can be generated from $t_j$ for $t_{j'}$. Without knowing $\mathcal{N}_i^{j'}$, user $t_j$ can estimate the value of $|\mathcal{N}_i^j\setminus \mathcal{N}_i^{j'}|$ based on its own received packets $\mathcal{N}_i^j$ as follows. Specifically,  $|\mathcal{N}_i^j\setminus \mathcal{N}_i^{j'}|=m$ if $m$ out of $|\mathcal{N}_i^j|$ received packets at user $t_j$ are erased at user $t_{j'}$, i.e.,
\begin{align}
&\Pr\left(|\mathcal{N}_i^j\setminus \mathcal{N}_i^{j'}|=m \Big |\mathcal{N}_i^j\right)\nonumber\\
&\quad =
\begin{cases}
{|\mathcal{N}_i^j|\choose m}p_1^m(1-p_1)^{(|\mathcal{N}_i^j|-m)}, & m\leq |\mathcal{N}_i^j|\\
0, & |\mathcal{N}_i^j|<m\leq M
\end{cases}\label{eq:probm}
\end{align}

For notational convenience, we will represent above conditional probability  as $\Pr(m|\mathcal{N}_i^j)$.  Assume that user $t_j$ has already sent out $u$ packets generated from batch $i$. Then, the $(u+1)$th packet generated from the same batch is still useful for user $t_{j'}$ if either  $m\geq u+1$ or  at most $(m-1)$ out of $u$ packets are received by $t_{j'}$. Denote this event by $E_i^u(j)$, its probability of occurrence be estimated as
{\small
\begin{align}
&\Pr(E_i^u(j)|\mathcal{N}_i^j)\nonumber\\
&=\sum_{m=u+1}^{M}\Pr(m|\mathcal{N}_i^j)+\sum_{m=1}^{u}\Pr(m|\mathcal{N}_i^j)\sum_{l=0}^{m-1}{u\choose l}(1-p_2)^lp_2^{(u-l)}.\label{eq:estimation}
\end{align}}

By symmetry, \eqref{eq:estimation}  applies for all the peers of user $t_j$. Hence, $\Pr(E_i^u(j)|\mathcal{N}_i^j)$ can be used as a metric to measure the usefulness for user $t_j$ to broadcast the $(u+1)$th packet generated from batch $i$.   In general, \eqref{eq:estimation} is valid for any  $u\geq 0$. However, a user usually will not send more than $M$ packets from the same batch before decoding. Thus we can  calculate the estimation only up to $u=M$. Let  $\mathbf{S}_j\in\mathbb{R}^{M\times n}$ be such  ``usefulness" matrix for user $t_j$, with the $(u,i)$th element equal to $\Pr(E_i^u(j)|\mathcal{N}_i^j)$. As more phase 2 packets are sent out from a batch,  new transmission is less likely to be useful. Hence, each column of $\mathbf{S}_j$ is a monotonically decreasing vector, i.e., $\mathbf{S}_j(u,i)>\mathbf{S}_j(u+1,i)$. Furthermore, if more packets are received for batch $i$ than batch $i'$ at the end of phase 1,  the packet generated from batch $i$ is more likely to be useful than that from batch $i'$, expressed mathematically, if we have $|\mathcal{N}_{i}^{j}|>|\mathcal{N}_{i'}^{j}|$, then $\mathbf{S}_{j}(u,i)>\mathbf{S}_j(u,i'), \forall u\geq 0$.

To maximize the spectral efficiency, a packet that is expected to be more useful should be transmitted with higher priority.  To obtain the optimal transmission order, user $t_j$ sorts all the elements in $\mathbf{S}_j$ in descending order. Denote the ordered elements by a vector $\mathbf{s}_j\in\mathbb{R}^{1\times Mn}$ and  the column index of the ordered elements by a vector $\mathbf v_j\in\mathbb{Z}^{1\times Mn}$. Then each element of $\mathbf{v}_j$ represents a batch ID within $\{1,...,n\}$. User $t_j$ sequentially transmits the temporally network-coded packets with batch ID obtained from $\mathbf v_j$.

\begin{example}
For  illustration purpose, we consider a simple setup with $p_0=0$, $p_1=0.5$, $p_2=0.1$, $M=4$ and $n=5$. If user $t_j$ receives $\{2,1,3,4,2\}$ packets at the end of phase 1 for batch $1$ to $5$, respectively, the usefulness matrix calculated based on \eqref{eq:estimation} is
\begin{align}
\mathbf{S}_j=\left[
\begin{matrix}
0.7500   & 0.5000  &  0.8750 &   0.9375  &  0.7500\\
0.3000   & 0.0500  &  0.5375  &  0.7125  &  0.3000\\
0.0525  &  0.0050  &  0.2000   & 0.3862 &   0.0525\\
0.0075  &  0.0005  &  0.0448  &  0.1410  &  0.0075\\
\end{matrix}
\right].
\end{align}
By sorting all the elements in $\mathbf{S}_j$, user $t_j$ obtains $\mathbf{s}_j=\left[\begin{matrix}0.9375 & 0.8750 & 0.7500 & 0.7500 & 0.7125 & \cdots \end{matrix}\right]$. The column indices of the elements in $\mathbf{s}_j$ give the transmission order as $\mathbf{v}_j=\left[\begin{matrix}4&3 &1 &5 &4 &3 &\cdots\end{matrix}\right]$. At the first time when user $t_j$ can access the channel, it will send a coded packet generated from all the available packets in its buffer for batch 4, received in both phases.  For the second transmission,  a coded packet from batch 3 will be sent out, and the process continues.
\end{example}

\subsubsection{Estimating the Total Number of Transmissions in Phase 2}
We assume that all the users have equal probability for transmission under a multiple-access scheme, such as TDMA or CSMA/CA. The transmission is complete when all the users are able to recover the file, which is assumed to be true after $T$ transmissions in total. On average, user $t_j$ will send out ${T}/{k}$ coded packets with batch IDs given by the first ${T}/{k}$  elements of $\mathbf v_j$. Among all the $T/k$ packets sent out by $t_j$, only $(1-p_2)T/k$ will reach user $t_{j'}$ due to packet erasures. Hence, the number of packets received by user $t_{j'}$ from its $(k-1)$ peers is a function of $T$, which is
\begin{align}
P(T)=\frac{(1-p_2)(k-1)T}{k}.\label{eq:PT}
 \end{align}

 By symmetry, we may assume that the batch IDs of the received packets follow a uniform distribution. In other words, if we denote by $\mathbf Y_2$ the total number of packets received for a typical batch in phase 2, then $\mathbf{Y}_2$ is a random number following binomial distribution $\mathcal{B}(P(T),{1}/{n})$, i.e.,
\begin{align}
&\Pr(\mathbf{Y}_2=i)={P(T)\choose i}\left(\frac{1}{n}\right)^i\left(1-\frac{1}{n}\right)^{P(T)-i},\quad i=0,...,P(T). \label{eq:PY2}
\end{align}

Denote by $\mathbf{Y}_1$ and $\mathbf{Z}$ the number of packets for a typical batch received by a single user and the  user group in phase 1, respectively. Clearly, we have $\mathbf{Y}_1\sim\mathcal{B}(M,(1-p_0)(1-p_1))$, i.e.,
\begin{align}
\Pr(\mathbf{Y}_1=i)={M\choose i}[(1-p_0)(1-p_1)]^{i}(p_0+p_1-p_0p_1)^{(M-i)}, i=0,...,M.\label{eq:PY1}
\end{align}
Furthermore, $\mathbf Z$ is a random variable which satisfies $\mathbf{Z}\geq\mathbf{Y}_1$, since the union of all sets is always larger than any single set. The distribution of $\mathbf Z$ given $\mathbf{Y}_1$ can be obtained as
\begin{align}
&\Pr(\mathbf{Z}=j|\mathbf Y_1=i)={{M-i}\choose {j-i}}(1-p_1^{k-1})^{j-i}p_1^{(k-1)(M-j)}, \forall j\geq i, i=0,...,M\label{eq:PZgY1}.
\end{align}

Since only $\mathbf Z$ packets are available for the whole user group, the number of useful packets available at any user cannot be larger than $\mathbf{Z}$. Any more packets received will be a linear combination of the existing $\mathbf{Z}$ packets. Out of these $\mathbf Z$ packets, the user already has $\mathbf{Y}_1$ packets  received during phase 1. Hence,  anything more than $(\mathbf{Z}-\mathbf{Y}_1)$ packets received during phase 2 will be redundant.  Furthermore, since random linear network coding over a sufficiently large field size is applied, we assume that any packet received before reaching its limit $\mathbf Z$ is innovative.
\begin{lemma}\label{lem:Delta}
Let $\mathbf{\Delta}=\mathbf{Z}-\mathbf{Y}_1$, we have  $\mathbf{\Delta}\sim\mathcal{B}(M,\tilde{p})$, where $\tilde{p}=(1-p_1^{k-1})(p_0+p_1-p_0p_1)$.
\end{lemma}
\begin{proof}
Please refer to Appendix~\ref{A:proofDelta}
\end{proof}

Based on Lemma~\ref{lem:Delta}, the expected number of redundant packets for all batches received during phase 2, denoted by $R(T)$,  can be estimated as:
\begin{align}
R(T)=n\sum_{l=\delta}^{P(T)}\sum_{\delta=0}^{M}(l-\delta)\Pr(\mathbf{Y}_2=l,\mathbf{\Delta}=\delta)\label{eq:redundancy}.
\end{align}
For simplicity, we assume that  $\mathbf{Y}_2$ and $\mathbf{\Delta}$ are \emph{independent} binomial random variables. Hence, $(\mathbf{Y}_2-\mathbf{\Delta})$ can be approximated as a Gaussian random variable distributed according to $\mathcal{N}(\mu_R,\sigma_R^2)$, where $\mu_R=\frac{P(T)}{n}-M\tilde{p}$ and $\sigma_R^2=\frac{P(T)}{n}\left(1-\frac{1}{n}\right)+M\tilde{p}(1-\tilde{p})$. Hence, \eqref{eq:redundancy} can be viewed as the positive expectation of a Gaussian variable, which can be computed as
\begin{align}
R(T)&=n\int_{0}^{\infty}\frac{x}{\sqrt{2\pi\sigma_R^2}}\exp\left(-\frac{(x-\mu_R)^2}{2\sigma_R}\right)dx\nonumber\\
&=n\sqrt{\frac{\sigma_R^2}{2\pi}}\exp\left(-\frac{\mu_R^2}{2\sigma_R^2}\right)+\mu_RnQ\left(-\frac{\mu_R}{\sigma_R}\right).\label{eq:redunExplicit}
\end{align}

Denote by $\mathbf{D}_j$ the number of innovative packets received by user $t_j$ at the end of phase 2, $j=1,...,k$. For simplicity, we assume $\mathbf{D}_j\sim\mathcal{N}(\mu_D,\sigma_D^2)$. Following similar analysis given above, we have
\begin{align}
\mu_D=(1-p_0)(1-p_1)nM+P(T)-R(T),\label{eq:muD}
\end{align}
where $P(T)$ is given in \eqref{eq:PT} and $R(T)$ is given in \eqref{eq:redunExplicit}. Since $R(T)$ is usually much smaller than $P(T)$ and the number of transmissions during phase 1, it can be ignored when computing $\sigma_D^2$. Hence, we have
\begin{align}
\sigma_D^2\approx nM(1-p_0)(1-p_1)(p_0+p_1-p_0p_1)+\frac{T(k-1)}{k}(1-p_2)p_2.\label{eq:varianceD}
\end{align}

Phase 2 is complete when the last user is able to decode the file, i.e.,
\begin{align}
\min\{\mathbf{D}_1,...,\mathbf{D}_k\}\geq (1+\eta)F. \label{eq:condition1}
\end{align}

Based on the approximation in \cite{Royston82}, \eqref{eq:condition1} can be explicitly expressed as
\begin{align}
\mu_D+\sigma_D\Phi^{-1}\left(\frac{0.625}{k+0.25}\right)\geq (1+\eta)F, \label{eq:determineT}
\end{align}
where $\Phi(\cdot)$ is the cumulative distribution function (cdf) of the standard normal distribution $\mathcal{N}(0,1)$; $\mu_D$ and $\sigma_D$ are functions of $T$, as given in \eqref{eq:muD} and \eqref{eq:varianceD}, respectively. Hence, we can obtain the stopping time $T$ from \eqref{eq:determineT}. Intuitively, the more packets transmitted, the more innovative packets will be received. Hence, the left-hand-side of \eqref{eq:determineT} is monotonically increasing with $T$; thus \eqref{eq:determineT} has a unique solution. The numerical solution of  \eqref{eq:determineT} can be efficiently obtained by bisection method, as shown in Algorithm~\ref{alg1}.
\begin{algorithm}
\caption{$T(n,M,k,F',p_0,p_1,p_2)$}
\label{alg1}
\begin{algorithmic}
\STATE{\textbf{Initialize:}}
\STATE{$\tilde{p}=(1-p_1^{k-1})(p_0+p_1-p_0p_1);\beta=\Phi^{-1}\left(\frac{0.625}{k+0.25}\right)$}
\STATE{$T_l=0$}
\STATE{$T_u=nM$}
\STATE{$f_l=\mu_D(T_l)+\sigma_D(T_l)\beta-F'$}
\STATE{$f_u=\mu_D(T_u)+\sigma_D(T_u)\beta-F'$}
\WHILE{$f_u-f_l>1$}
\STATE{$T=\frac{T_l+T_u}{2}$}
\STATE{$f=\mu_D(T)+\sigma_D(T)\beta-F'$}
\IF{$f>0$}
\STATE{$T_u=T$;$f_u=f$}
\ELSE
\STATE{$T_l=T$;$f_l=f$}
\ENDIF
\ENDWHILE
\STATE{$T=T_u$}
\end{algorithmic}
\end{algorithm}

\subsubsection{Estimating the Rank Distribution}
In conventional directed \emph{acyclic} networks considered in \cite{Yang11,Yang13}, the rank distribution of the batches is determined by the network topology and the erasure probability of each link. However, for the network shown in Fig.~\ref{F:model} with cycles, the rank distribution is affected by the scheduling scheme in phase 2 transmission.  Since rank distribution is an important parameter for designing BATS code, it is crucial to get a good estimation for it before transmission, which is pursued in this subsection.

Following similar assumptions as in \eqref{eq:redundancy}, the rank of a typical batch for a user is $r$ if either of the following two events occur: i) the user group has more than $r$ packets for this batch, but the user only receives $r$, i.e., $\mathbf Z>r$ and $\mathbf{Y}_1+\mathbf{Y}_2=r$; ii) the user receives more than $r$ packets, but only $r$ out of them are innovative, i.e., $\mathbf{Y}_1+\mathbf{Y}_2\geq r$ and $\mathbf{Z}=r$.  Hence, at the end of the transmissions, the probability that a batch has rank $r, r\in\{0,1,...,M\}$, for a user is given by
\begin{align}
&\Pr(r)=\Pr(\mathbf{Z}>r,\mathbf Y_1+\mathbf Y_2=r)+\Pr(\mathbf{Z}=r,\mathbf Y_1+\mathbf Y_2\geq r)\nonumber\\
&=\sum_{i=0}^{r}\sum_{j=r+1}^{M}\Pr(\mathbf Z=j|\mathbf Y_1=i)\Pr(\mathbf Y_1=i)\Pr(\mathbf{Y}_2=r-i)\nonumber\\
&\quad +\sum_{i=0}^{r}\sum_{j=r-i}^{P}\Pr\left(\mathbf Z=r|\mathbf{Y}_1=i\right)\Pr(\mathbf Y_1=i)\Pr(\mathbf Y_2=j),\label{eq:Rank}
\end{align}
where each individual probability $\Pr(\mathbf{Y}_2)$, $\Pr(\mathbf{Y}_1)$ and $\Pr(\mathbf Z|\mathbf Y_1)$ are given in \eqref{eq:PY2},\eqref{eq:PY1} and \eqref{eq:PZgY1}, respectively, with $T$ obtained from Algorithm~\ref{alg1}.

Since the number of source transmissions is set to the minimum, the decoding usually occurs only when $T$ is sufficiently large. In this case, we may assume that $\mathbf{Y}_1+\mathbf{Y}_2>\mathbf{Z}$ and $\Pr(r)\approx\Pr(\mathbf{Z}=r)$, i.e.,
\begin{align}
\Pr(r)={M\choose r}(1-\bar{p})^{r}\bar{p}^{M-r},\label{eq:rankApprox}
\end{align}
where $\bar{p}=p_0+p_1^{k}-p_0p_1^{k}$.

\begin{example}
Assume that a file containing 1600 packets is to be transmitted from the source to three users through the network shown in Fig.~\ref{F:model} with $p_0=0.05$, $p_1=0.5$ and $p_2=0.1$. A batch code with batch size $M=16$ is used to correct the erasures. The batch overhead is assumed to be $1\%$ and hence the minimum number of batches sent by the source is computed from \eqref{eq:nValue} to be 129. The analytical rank distribution is plotted together with the simulated rank distribution for the three users in Fig.~\ref{F:RankDist}. It is observed that the analytical rank distribution given in \eqref{eq:Rank} matches quite well with the simulation results. Furthermore, the approximated rank distribution given in \eqref{eq:rankApprox} is also of sufficient accuracy.  Hence, we can design good BATS code based on the rank distribution given in \eqref{eq:Rank} or \eqref{eq:rankApprox}.
\begin{figure}[htb]
\centering
\includegraphics[scale=0.6]{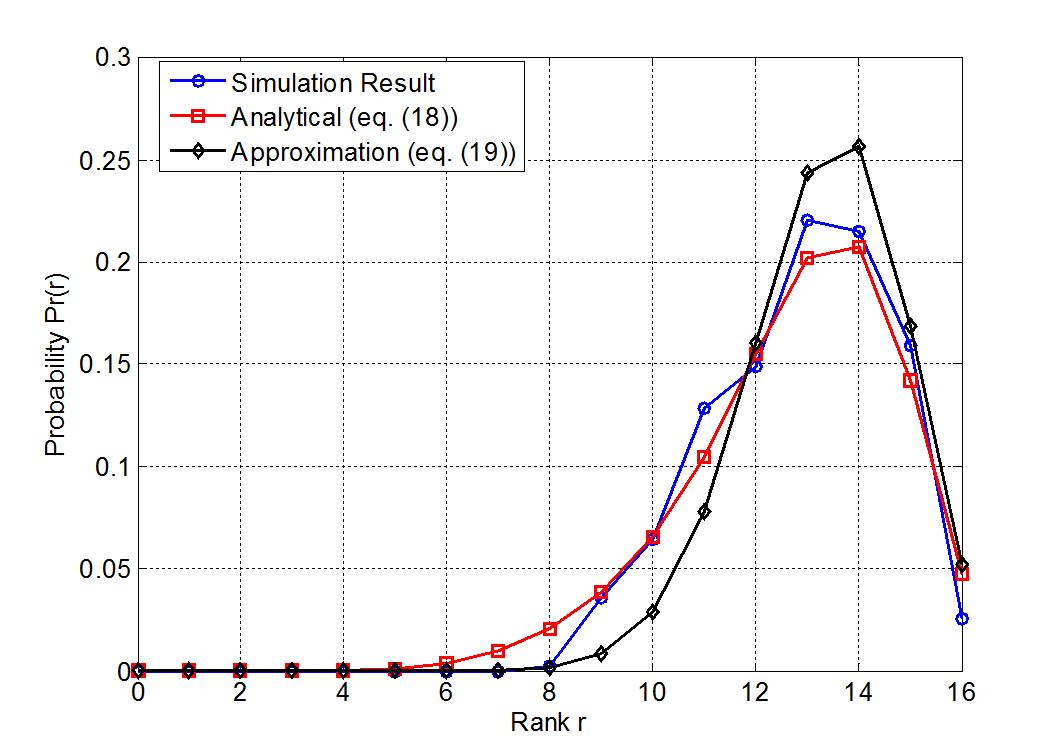}
\caption{Comparison between analytical and simulated rank distributions. }
\label{F:RankDist}
\end{figure}
\end{example}

\subsection{Minimizing the total number of transmissions}\label{sec:nopt}
It is observed that if the source sends a few more batches than the threshold given in \eqref{eq:nValue}, the total number of transmissions in both phases may be significantly reduced. The optimal number of batches to be transmitted by the source, denoted as $n^*$, can be found by solving the following optimization problem
\begin{equation}\label{eq:optimization}
\begin{aligned}
n^*=&\arg\underset{n}{\min} (T+nM)\\
&\textnormal{ subject to: \eqref{eq:redunExplicit}-\eqref{eq:varianceD},\eqref{eq:determineT}}\\
\end{aligned}
\end{equation}

Since the constraints in \eqref{eq:optimization} is non-convex, a closed-form solution does not exist in general. We propose to solve for $n^*$ by exhaustively searching all possible values of $n$ since the search space is not large, as explained next.  First, the minimum value of $n$, denoted by $n_l$, is given by \eqref{eq:nValue}. Furthermore, the maximum value of $n$, denoted by $n_u$, is the number of batches sent out by the source when all the users are able to decode the file without requiring phase 2 transmissions. To find $n_u$, we further denote by $\mathbf{X}_j$ the number of packets received by user $t_j$, for $j=1,...,k$, after the source sent out $n_u$ batches, i.e., $Mn_u$ packets.  $\mathbf{X}_j$ are independent and identically  distributed random variables according to $\mathcal{B}(Mn_u,(1-p_0)(1-p_1))$.  The source transmission stops when all the users are able to recover the file, i.e., when
\begin{align}
\min_{j\in\{1,...,k\}}\{\mathbf{X}_j\}\geq (1+\eta)F. \label{eq:condition}
\end{align}
Since $Mn_u\gg 1$, the binomial distribution can be approximated by the normal distribution $\mathcal{N}(\bar{\mu},\bar{\sigma}^2)$, where $\bar{\mu}=Mn_u(1-p_0)(1-p_1)$ and $\bar{\sigma}^2=Mn_u(p_0+p_1-p_0p_1)(1-p_0)(1-p_1)$. The statistical mean of $\min\{\mathbf{X}_j\}$ can be approximated as \cite{Royston82}
\begin{align}
\mathbb{E}\left[\min_{j\in\{1,...,k\}}\{\mathbf{X}_j\}\right]\approx \bar{\mu}+\bar{\sigma}\Phi^{-1}\left(\frac{0.625}{k+0.25}\right), \label{eq:approx}
\end{align}
 By substituting \eqref{eq:approx} into \eqref{eq:condition}, we can solve for $n_u$ as
\begin{align}
n_u=\left\lceil\frac{2F'+\hat{p}\beta^2+\sqrt{4\hat{p}\beta^2F'+\hat{p}\beta^4}}{2M(1-\hat{p})}\right\rceil, \label{eq:nmax}
\end{align}
where $F'=(1+\eta)F$, $\hat{p}=p_0+p_1-p_0p_1$ and $\beta=\Phi^{-1}\left(\frac{0.625}{k+0.25}\right)$.

For all the integer values within $[n_l,n_u]$, the corresponding phase 2 transmissions $T$ can be found from Algorithm~\ref{alg1}, and hence the total number of transmissions $nM+T$ can be computed. The value of $n$ that leads to the minimum number of transmissions is returned as $n^*$.

\begin{example}\label{ex:5user}
If a file containing $5000$ packets is to be sent to a group of $k=5$ receivers through the network shown in Fig.~\ref{F:model}, where $p_0=0.05$, $p_1=0.5$, $p_2=0.1$. Assume that a BATS code with batch size $M=16$ is used and the degree distribution is well designed such that the coding overhead is maintained within $1\%$. From \eqref{eq:nValue} and \eqref{eq:nmax}, the minimum and maximum value of $n$ can be computed to be $351$ and $673$, respectively. For all the integers in between, we can find the corresponding total number of transmissions required, which is plotted in Fig.~\ref{F:nopt}, and the the optimal number of batches is $n^*=402$.
\begin{figure}[htb]
\centering
\includegraphics[scale=0.6]{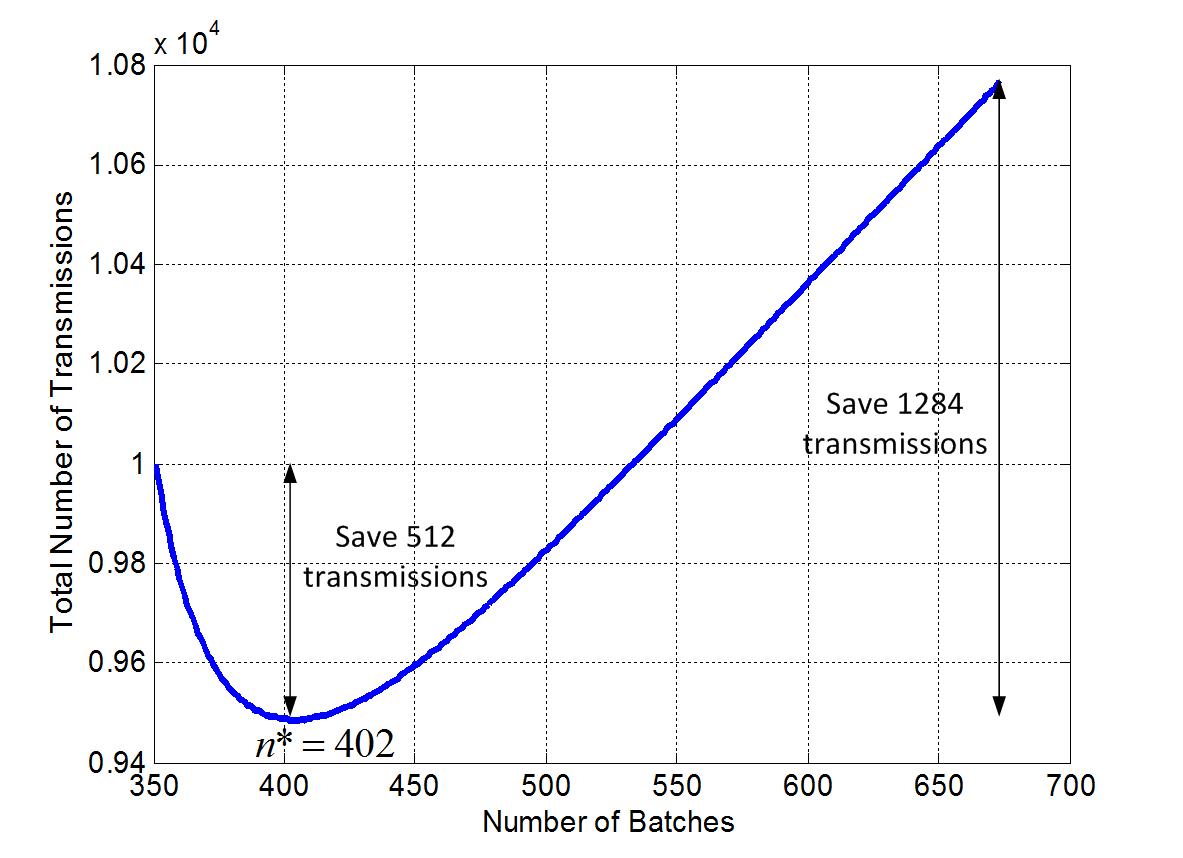}
\caption{Total number of transmissions (phase 1 and 2) versus number of batches.}
\label{F:nopt}
\end{figure}

Compared with the single phase transmissions where the number of batches sent by the source is set to the maximum value $n_u$, the proposed cooperative broadcast scheme with $n^*$ batches saves 1284 transmissions. Furthermore, compared with the one targeting at minimum source transmissions, i.e., with $n=351$, 512 transmissions are saved by making the source transmit 51 more batches with $n^*=402$.
\end{example}

When the number of batches sent is larger than the minimum value, the users do not have to receive all the innovative packets from its peers before decoding the file. Hence, the resulting channel rank distribution will not follow that given in \eqref{eq:rankApprox}. However, the estimated distribution given in \eqref{eq:Rank} still holds. Hence, we can find the optimal degree distribution for the BATS code accordingly.

\begin{example}
Consider the same network settings as that in Example~\ref{ex:5user}. The analytic channel rank distribution given in \eqref{eq:Rank} is compared with simulated distribution in Fig~\ref{F:rank5user}. The simulated rank distribution is averaged over the 5 users at the moment when all the users can recover the file. Eventually, the estimated distribution matches well with the simulated one. The slight discrepancy for high-rank batches is mainly due to the assumption that all batches are sent with equal probability in phase 2. The experimental results in Section~\ref{sec:experiment} show that the performance degradation of the BATS code due to this small error of channel estimation is negligible.
\begin{figure}[htb]
\centering
\includegraphics[scale=0.6]{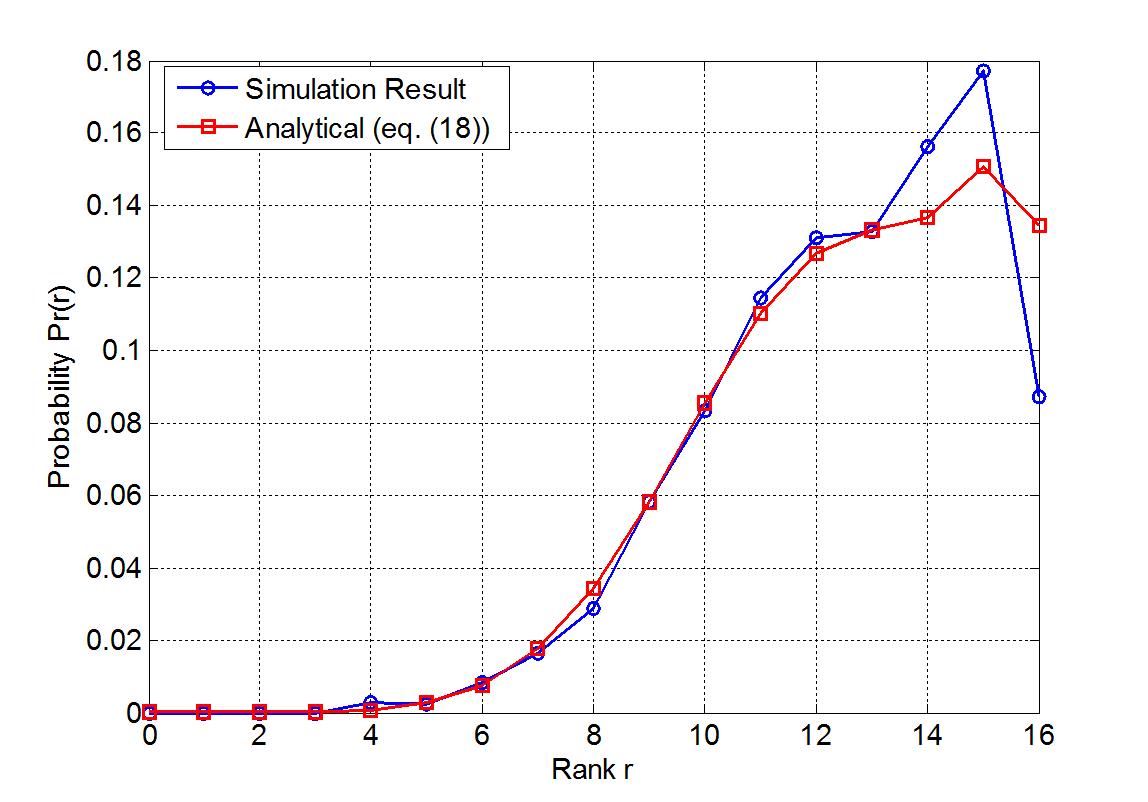}
\caption{Comparison between the analytical and simulated channel rank distributions.}
\label{F:rank5user}
\end{figure}
\end{example}

\section{Performance Evaluation}\label{sec:performance}
This section is devoted to evaluating the effectiveness, efficiency, and robustness of the proposed 2-phase cooperative broadcasting schemes. First, the transmission efficiency of the proposed 2-phase scheme with optimal number of batches is compared with the coded single-phase transmissions and the cooperative P2P information exchange (PIE) introduced in \cite{Fan2009,Fan2010}. Then, the computational overhead of the proposed scheme is analyzed. Finally, we also investigate the robustness of the proposed scheme against unknown number of users.

\subsection{Transmission Efficiency}
\subsubsection{Comparison with Single-Phase Transmission}
First, we compare the proposed two-phase broadcasting scheme with the traditional single-phase transmission, where the source keeps transmitting the packets until all the users can recover the file. Assume that optimal erasure code, such as Raptor code \cite{Shokrollahi11}, has been applied so that a user can recover the file after receiving $(1+\eta)F$ packets.


 Note that the proposed two-phase scheme in Section~\ref{sec:nopt} reduces to the coded single-phase transmission when $n^*$ obtained from \eqref{eq:optimization} is equal to the upper bound, $n_u$, given in \eqref{eq:nmax}.
Moreover, the proposed scheme should outperform the single-phase broadcast if the inter-user links are  better than the links between the source and the users, and/or the number of users is sufficiently large.

\begin{example}
Assume a file containing $F=2000$ packets is to be distributed by the source node to a group of $k$ users. The coding overhead $\eta$ is set to be $1\%$ in for both BATS code and the erasure code. The total number of transmissions required for all users to recover the file is plotted against the number of users for both the single-phase transmission and the proposed two-phase broadcast scheme with optimal number of batches.

\begin{figure}[htb]
\centering
\includegraphics[scale=0.6]{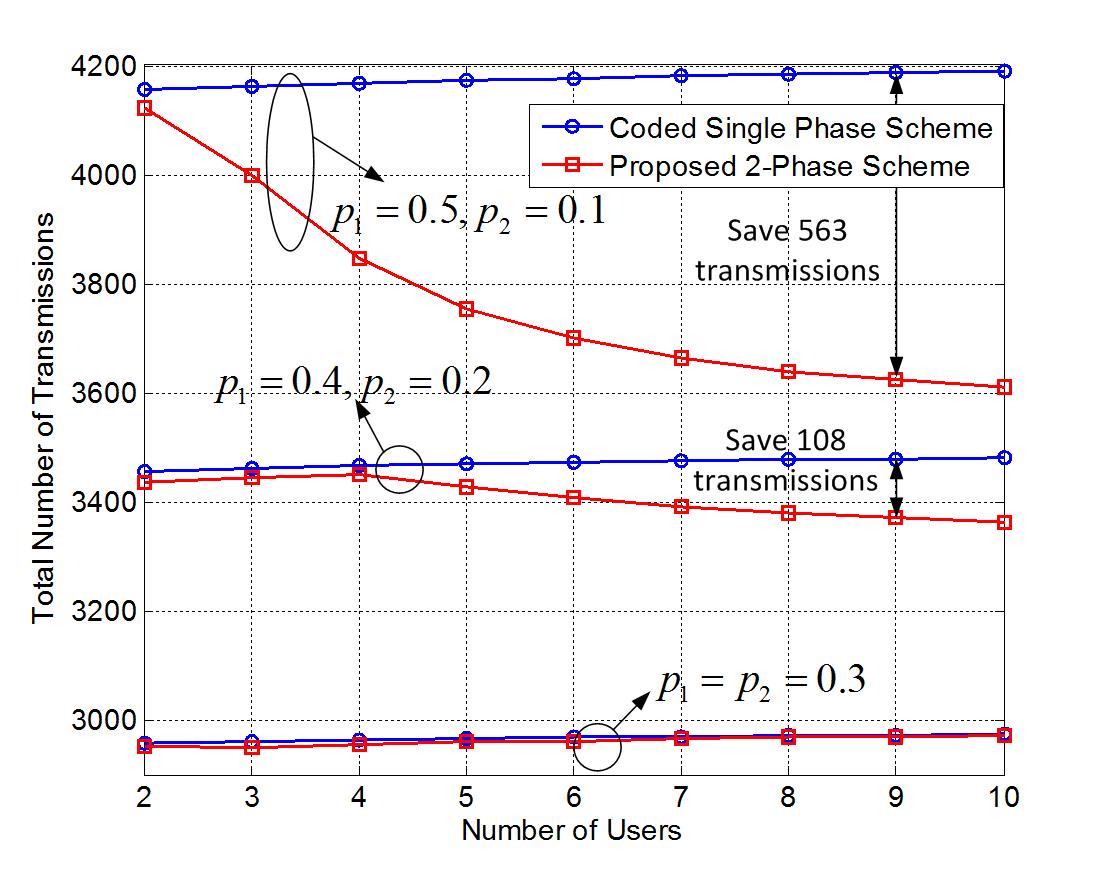}
\caption{Comparison of the proposed two-phase cooperative broadcast with single-phase transmission.}
\label{F:Compare}
\end{figure}

It is observed from Fig.~\ref{F:Compare} that the total number of transmissions increases with the number of users in the single phase scheme. In contrast, with the proposed two-phase scheme, it decreases with the number of users due to spatial diversity gain.  Furthermore, when the inter-user links have the same packet loss rate as the links between the source and the users, i.e., $p_1=p_2$, the proposed two-phase cooperative broadcast  reduces to the single-phase transmission as the optimal number of batches $n^*$ obtained from \eqref{eq:optimization} equals to $n_u$ in \eqref{eq:nmax}. When $p_2<p_1$, the proposed two-phase scheme outperforms the single-phase transmission, e.g., with $k=9$, the proposed scheme saves 108 transmissions when $p_1=0.4, p_2=0.2$, and it saves 563 transmissions when $p_1=0.5,p_2=0.1$.
\end{example}

\subsubsection{Comparison with cooperative peer-to-peer information exchange (PIE) \cite{Fan2010}}
PIE is an efficient  peer scheduling algorithm for network coding enabled wireless networks, introduced in \cite{Fan2009,Fan2010}. However, phase 1 transmission in PIE is uncoded, and hence the residual loss is inevitable. To make a fair comparison, we assume that the file is firstly encoded with some good erasure code so that the receiver can recover the file upon receiving $(1+\eta)F$ packets. The coded packets is then divided into blocks of size $M$ and sequentially broadcasted by the source node. Furthermore, since PIE is designed only for lossless inter-user channels, when there are erasures in phase 2, i.e., $p_2>0$, the scheduling algorithm will terminate before the the data block can be decoded. To achieve reliable communications in phase 2, additional retransmissions are required, which results in some performance degradation. In the following example, we compare the performance of the proposed cooperative two-phase scheme with PIE, in favor of the latter by ignoring transmission overhead used for  exchanging the state information.
\begin{example}
Consider the network shown in Fig.~\ref{F:model}, where the source node is intended to send a file with $F=2000$ packets to a group of $k$ users. The links between the source node with the users are assumed to have independent erasures with probability $p_1$.

The total number of transmissions required for all the users to recover the file with PIE and with the proposed scheme are compared in Fig.~\ref{F:PIECompare}(a), under the assumption of lossless phase 2 links, i.e., $p_2=0$. It is observed that the proposed scheme outperforms PIE when $p_1=0.2$. This is consistent with expectation because complete repair is not necessary with proposed scheme in phase 2, due to the application of BATS code, but when a packet is missed by in PIE, it must cost a retransmission to repair. When $p_1=0.5$, the received packets among different users have more diverse erasures, hence PIE, which has a centralized scheduling algorithm, slightly outperforms the proposed scheme, at the cost of additional control complexity and state information exchange. When there are erasures in phase 2 channels, i.e., $p_2=0.1$, the proposed scheme outperforms PIE for both $p_1=0.2$ and $p_1=0.5$, as shown in Fig.~\ref{F:PIECompare}(b).

\begin{figure}[htb]
\centering
\includegraphics[scale=1]{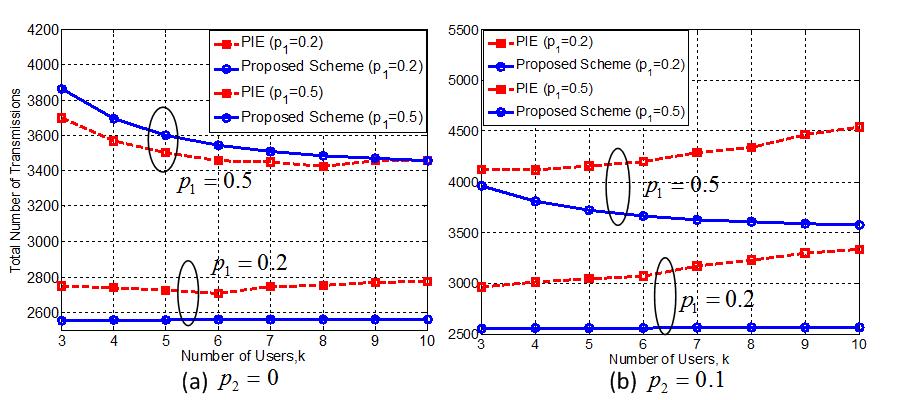}
\caption{Comparison of the proposed scheme with PIE.}
\label{F:PIECompare}
\end{figure}

\end{example}

\subsection{Computational Overhead}
The design of the proposed 2-phase cooperative broadcasting scheme involves determining the number of batches, the degree distribution of the BATS code, scheduling of each users and BATS decoding. First, the minimum number of batches can be obtained from equation \eqref{eq:nValue} directly and the optimal number of batches is determined by solving the optimization in \eqref{eq:optimization}, which is of complexity $\mathcal{O}(\frac{K}{M}\log_2K)$. When solving for the number of batches, we can also obtain the estimated number of transmissions of phase 2. Then, the corresponding channel rank distribution $\mathbf h$ can be directly computed from \eqref{eq:Rank}. Based on $\mathbf h$, the optimal degree distribution for the BATS code can be obtained by solving a linear optimization problem formulated in \cite{Yang13}. All these computations can be carried out off-line, which will not cause any communication delay. On-the-fly computations includes the scheduling and BATS decoding. The proposed scheduling algorithm is completely distributed, which consists of computing and sorting the usefulness matrix, with complexity $\mathcal{O}(Mn\log M)$. BATS decoding is based on belief propagation and inactivation decoding, with complexity $\mathcal{O}(K(M^2+ML))$  \cite{Yang11}, where $L$ is the packet length.

\subsection{Robustness}
In certain scenarios,  the number of users $k$ may be unknown. The existing P2P repair schemes in \cite{Sanigepalli2006,BOPPER,Cheung2006,CPR,DCPR,Legout2006,Fan2010,CDS,CDSextension,XORRepair}  are not applicable for such scenarios since they require the state information of all the users for designing the scheduling algorithms. In contrast, the proposed two-phase cooperative broadcasting scheme is fully distributed, hence applicable for the case of unknown number of users.

 Due to increasing diversity gain with $k$, the proposed scheme, designed for a network with $k$ users, allow a network with unknown extra users to recover the file, at the cost of some performance degradation. Denote by $L(k)$ the minimum number of transmissions required to deliver a file from the source to a group of $k$ users with the proposed two-phase protocol. Under the same setup, another network with $k'$ users, where $k'\geq k$, should also be able to recover the file after $L(k)$ transmissions because additional users bring more space diversity. In other words, the performance degradation can be bounded  by the difference between the minimum number of transmissions, i.e., $L(k)-L(k')$.

Since the batches are chosen with equal probability in phase 2, with fixed number of batches, the channel rank distributions at the moment of decoding are almost the same for different number of users. Hence, the redundant transmissions caused by extra users in the second network is mainly due to the sub-optimal choice of $n$, which is much lower than the bound $L(k)-L(k')$.

\begin{example}
Consider the network shown in Fig.~\ref{F:model}, with $p_0=0.05$, $p_1=0.5$ and $p_2=0.1$. A file containing $K=2083$ packets is to be transmitted to $k$ users with batch size $16$.  If the proposed scheme is designed for $k=3$, the optimal number of batches is set as $n^*=200$ and the degree distribution of BATS code is designed based on the estimated channel rank distribution for $k=3$. In Fig.~\ref{F:changek}, the number of required transmissions for networks with $k\geq 3$ (based on the two-phase scheme designed for $k=3$) is compared with the ideal case (where the proposed two-phase cooperative broadcast scheme is designed for the exact $k$ users). It is observed that the performance degradation increases when $k$ gets larger, as expected. However, the degradation is still marginal in proportion, for example, the degradation is less than $5\%$ even when the networks has 3 times more users than it was designed for.

\begin{figure}[htb]
\centering
\includegraphics[scale=0.6]{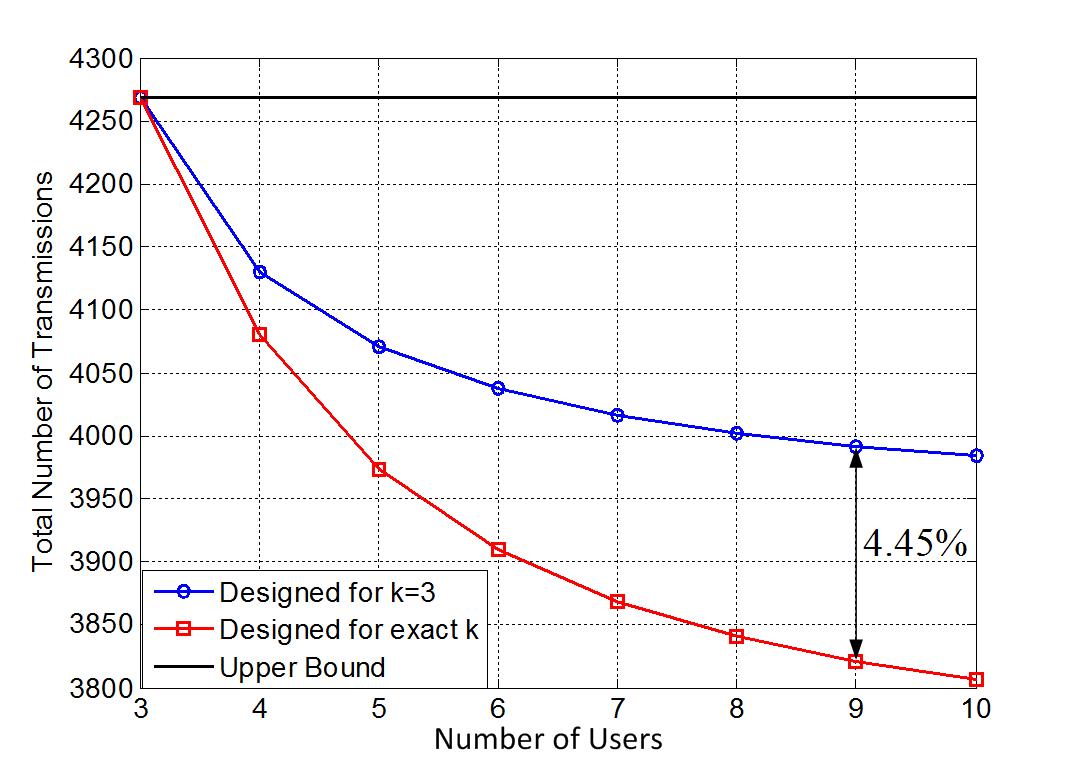}
\caption{Minimum number of transmissions versus number of users $k$. }
\label{F:changek}
\end{figure}

\end{example}

On the other hand, if the proposed scheme is designed for a $k$ larger than the actual value, the users may not be able to recover the file. Hence, when the exact value of $k$ is unknown, the proposed two-phase scheme should be designed for the minimum expected value.

\section{Experimental Results}\label{sec:experiment}
 This section evaluates the performance of the proposed two-phase broadcast protocol over a 4-node testbed based on IEEE802.11g wireless network, in order to validate the analytical results in Section~\ref{sec:main}. The testbed consists of 3 laptops as receivers and one desktop as the source. We use the HP ProBook 430G1 laptops which has inbuilt Wi-Fi and HP Z210 desktop computer which uses PROLiNK USB Wi-Fi dongle WG2000/R. The operating systems used in the laptops and the desktop  are Ubuntu 13.10 and Ubuntu 12.04, respectively. The source and receivers are connected in Ad-Hoc mode.

A picture file of size 2.1MB is distributed from the source to all the three receivers. The file is divided into 2083 packets, each of 1000 bytes. In phase 1, the source broadcast the coded packets to all its receivers, while in phase 2 the receivers exchange network-coded packets using broadcast transmission. Since the transmissions are carried out in broadcast mode, UDP is used as the transport layer protocol.

\subsection{Channel Characterization}
The IEEE 802.11 MAC operates in two modes, namely unicast and broadcast. Our testbed operates in the broadcast mode. There are two inherent problems in this mode: poor reliability and lack of back-off \cite{Katti2008}. Since there are multiple receivers in broadcast communication, it is unclear who should ACK. In the absence of ACK, there will be no retransmission. Furthermore, the broadcast sender cannot sense the medium and it will keep transmitting during collision without backing off, causing collision. These lead to packet loss at the receivers.

There are two kinds of packet loss in 802.11 broadcast: correlated loss and uncorrelated loss. The correlated loss refers to the common loss experienced by all users, which is mainly caused by collisions. On the other hand, the uncorrelated loss refers to the independent packet erasures at the receivers, which is mainly caused by interference and noise \cite{WLAN}.

\subsection{Experimental Results}
The correlated erasure probability for phase 1 is measured to be $5\%$, i.e., $p_0=0.05$. The uncorrelated loss\footnote{The instant packet loss rate changes over time due to interference and changing environment. Hence, the average packet loss rate over the transmission of entire file is used.} is set to $0.5$, i.e., $p_1=0.5$. With network coding, most of the packets sent out during phase 2 transmissions are innovative for the corresponding receivers, unless it reachs the limit determined by the total number of received packets in phase 1. Each user is able to recover the file upon receiving a sufficient number of innovative packets. Hence, it is unnecessary to differentiate the correlated and independent packet loss in phase 2. The erasure probability for phase 2 is measured to be around $p_2=0.2$, which is mainly due to congestion\footnote{802.11 broadcast mode lacks congestion control mechanism. The erasure probability for phase 2 transmission can be reduced if some congestion control mechanism, such as the pseudo-broadcast in \cite{Katti2008}, is applied.}.  A BATS code over $GF(2^8)$ is used to encoded the packets into batches of size $M=16$. The analytical channel rank distributions are derived with BATS code overhead set as $\eta=1\%$.

Based on the analysis presented in Section~\ref{sec:main}, the minimum number of batches sent by the source is $167$. To minimize the total number of transmissions, the optimal number of batches sent by the source is $211$. The number of innovative packets received is plotted against the number of transmissions in Fig.~\ref{F:n167} and Fig.~\ref{F:n211} for $n=167$ and $n=211$, respectively. Note that a packet is called innovative, if it is not a linear combination of the previously received packets. On the other hand, if a received packet is a linear combination of the previous packets within the batch, this packet is viewed as redundant packet.
\begin{figure}[htb]
\centering
\includegraphics[scale=0.6]{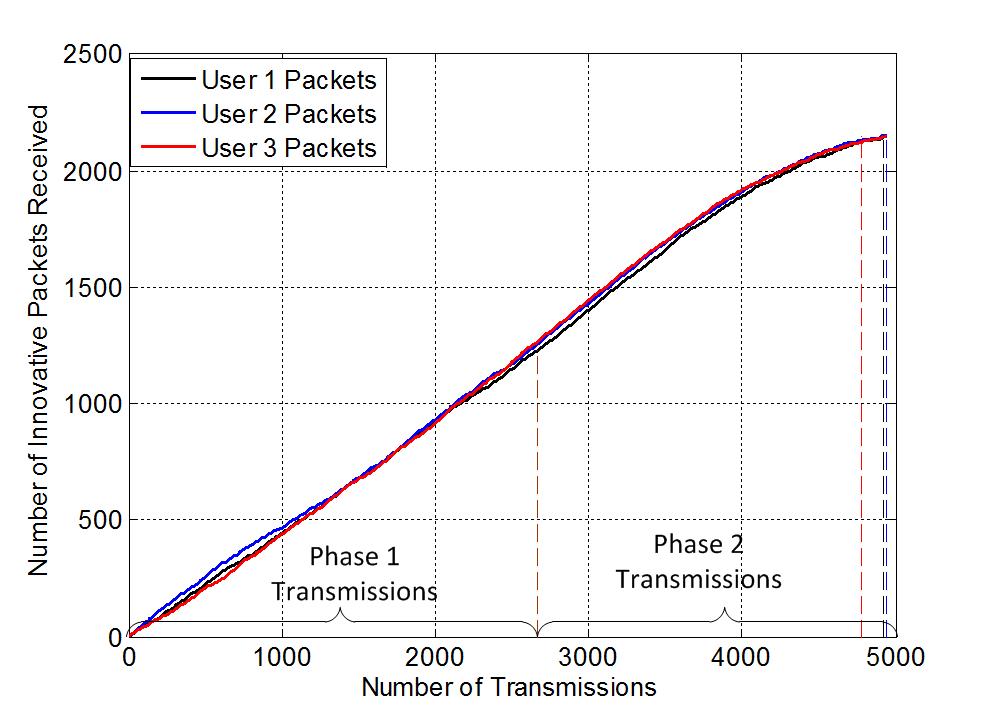}
\caption{Number of innovative packets received versus number of transmissions for $n=167$.}
\label{F:n167}
\end{figure}

\begin{figure}[htb]
\centering
\includegraphics[scale=0.6]{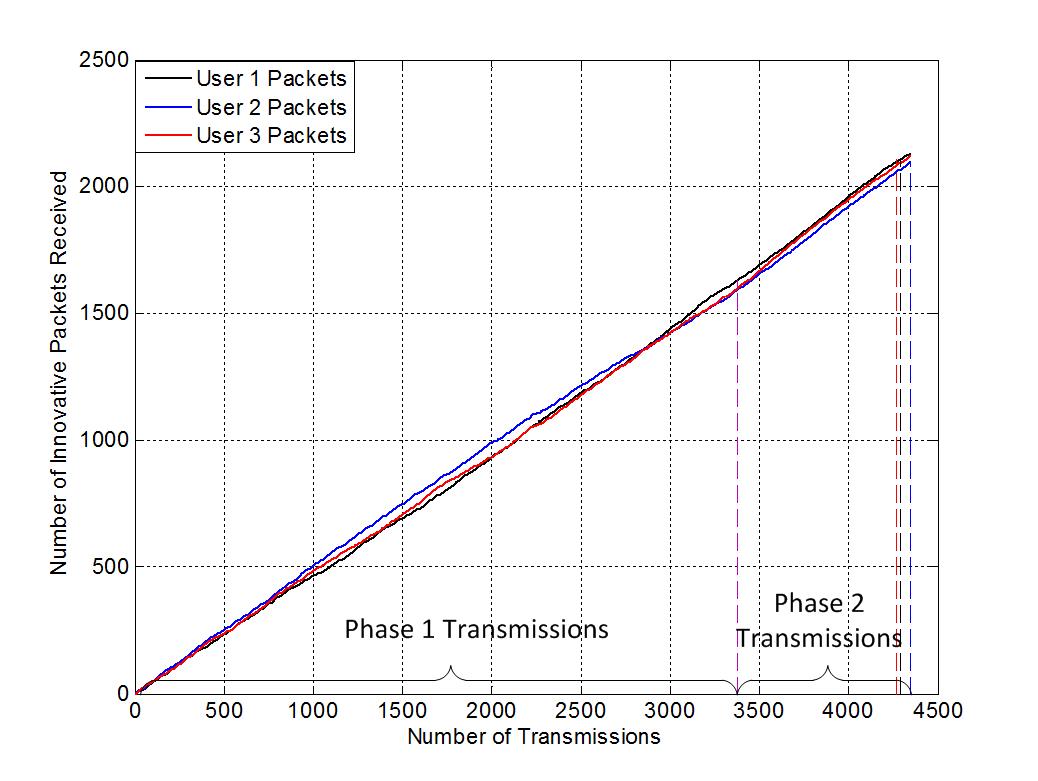}
\caption{Number of innovative packets received versus number of transmissions for $n=211$.}
\label{F:n211}
\end{figure}

For the case with $n=167$, user 1, user 2 and user 3 recover the file after receiving 2432, 2455 and 2395 packets, respectively. The total number of transmissions made by the network in both phases is 4939. The number of redundant packets received is measured to be 340 on average, which is $15\%$ of the total number of transmissions in phase 2. The redundance is close to the estimated value 325 computed from \eqref{eq:redundancy}.  The number of innovative packets used for decoding by the three users are 2091, 2088 and 2083, respectively, which means that the overhead of the BATS code is maintained within $0.4\%$. This small overhead validates our channel rank estimation made in \eqref{eq:Rank}, based on which the BATS code is designed.

For the case with $n=211$, the number of phase 2 transmissions is measured to be 968, which is close to the estimated value 956. The redundant packets received in phase 2 is 13 on average, which is  $1.34\%$ of all phase 2 transmissions. Furthermore, the file is recovered from 2091, 2086 and 2090 packets at user 1, user 2 and user 3, respectively. Hence, the coding overhead for the BATS code is maintained within $0.5\%$. The small overhead of BATS code designed based on the analytical channel rank distribution further validates our analytical results presented in Section~\ref{sec:main}. Compared with the case with $n=167$, 704 more packets are transmitted by the source in phase 1. However, the total number of transmissions in both phases is reduced from 4939 to 4344.


\section{Conclusion}\label{sec:conclusion}
In this paper, we have proposed a fully distributed two-phase cooperative broadcasting scheme based on BATS code to achieve reliable communication from the source node to a group of users. With the proposed scheme, the number of source transmissions is reduced by introducing user cooperations in phase 2. Furthermore,  the total number of retransmissions may also be reduced when the inter-user channels exploited in phase 2 have less erasures than the phase 1 channel from the source to the users. The performance of the proposed two-phase scheme has been analyzed and validated through simulations and experiments. When the power or bandwidth at the source is limited, we propose to apply the two-phase cooperative broadcast scheme with minimum number of batches. Otherwise, the proposed scheme with the optimal number of batches should be applied to minimize the total number of transmissions, and hence the communication delay. When global state information is available, the proposed two-phase protocol can be further improved by optimizing its scheduling algorithm.

\appendices
\section{Proof of Lemma \ref{lem:Delta}}\label{A:proofDelta}
Since $\mathbf{\Delta}=\mathbf{Z}-\mathbf{Y}_1$, its probability distribution can be calculated as
\begin{align}
\Pr(\mathbf{\Delta}=\delta)&=\sum_{i=0}^{M-\delta}\Pr(\mathbf{Z}=i+\delta|\mathbf{Y}_1=i)\nonumber\\
&=\sum_{i=0}^{M-\delta}{{M-i} \choose \delta}(1-p_1^{k-1})^{\delta}p_1^{(k-1)(M-i-\delta)}{M\choose i}\left[(1-p_0)(1-p_1)\right]^i(p_0+p_1-p_0p_1)^{M-i}\label{eq:before}
\end{align}
For notational convenience, denote $\hat{p}=p_0+p_1-p_0p_1$ and $\tilde{p}=(1-p_1^{k-1})\hat{p}$. Then, \eqref{eq:before} can be expressed as
\begin{align}
\Pr(\mathbf{\Delta}=\delta)&=\sum_{i=0}^{M-\delta}{{M-i} \choose \delta}\tilde{p}^{\delta}p_1^{(k-1)(M-i-\delta)}{M\choose i}(1-\hat{p})^i\hat{p}^{M-\delta-i}\nonumber\\
&={M\choose \delta}\tilde{p}^{\delta}
\sum_{i=0}^{M-\delta}{{M-\delta}\choose i}p_1^{(k-1)(M-i-\delta)}(1-\hat{p})^i\hat{p}^{M-\delta-i}\nonumber\\
&={M\choose \delta}\tilde{p}^{\delta}(1-\tilde{p})^{M-\delta}
\sum_{i=0}^{M-\delta}{{M-\delta}\choose i}\frac{p_1^{(k-1)(M-i-\delta)}(1-\hat{p})^i\hat{p}^{M-\delta-i}}{(1-\hat{p}+p_1^{k-1}\hat{p})^{M-\delta}}\nonumber\\
&={M\choose \delta}\tilde{p}^{\delta}(1-\tilde{p})^{M-\delta}\sum_{i=0}^{M-\delta}{{M-\delta}\choose i}\left(\frac{1-\hat{p}}{1-\hat{p}+p_1^{k-1}\hat{p}}\right)^i\left(\frac{p_1^{k-1}\hat{p}}{1-\hat{p}+p_1^{k-1}\hat{p}}\right)^{M-\delta-i}\nonumber\\
&={M\choose \delta}\tilde{p}^{\delta}(1-\tilde{p})^{M-\delta}.
\end{align}
This completes the proof of Lemma~\ref{lem:Delta}.

\bibliographystyle{ieeetr}
\bibliography{BATS}

\end{document}